\documentclass[a4paper,12pt]{article}
\usepackage{amsmath,amsfonts,amsthm}
 
\def\xhat{\hat x}
\def\That{\widehat T}
\def\numCP{P}
\def\CPlength{\lambda}
\def\numT{C}
\let\rho=\varrho
\let\epsilon=\varepsilon
\def\hide#1{}
\def\calM{{\mathcal{M}}}

\newtheorem{theorem}{Theorem}
\newtheorem{lemma}[theorem]{Lemma}

\newtheorem{observation}[theorem]{Observation}
\newtheorem{proposition}[theorem]{Proposition}

\begin{document}
\title{The Mixing Time of Glauber Dynamics
	for Colouring Regular Trees\footnote{
Partly funded by the EPSRC
grants   EP/E062482/1 and EP/E064906/1
}}
\author{Leslie Ann Goldberg\footnote{
Department of Computer Science, University of Liverpool, Liverpool L69 3BX, UK
}, 
Mark Jerrum\footnote{School of Mathematical Sciences, Queen Mary, University of London, Mile End Road, London E1 4NS, UK} 
and Marek Karpinski\footnote{
Department of Computer Science,
University of Bonn,
R\"omerstr. 164,
53117 Bonn, Germany
}}
\date{\today}
\maketitle

\begin{abstract}
We consider Metropolis Glauber dynamics for sampling
proper $q$-colourings of the $n$-vertex complete $b$-ary tree  
when $3\leq q\leq b/2\ln(b)$.
We give both upper and lower bounds
on the mixing time. 
For fixed $q$ and $b$, our upper bound is $n^{O(b/\log b)}$ 
and our lower bound is $n^{\Omega(b/q \log(b))}$, where 
the constants implicit in the $O()$ and $\Omega()$ notation do not depend upon~$n$, $q$ or~$b$.
\end{abstract}

\section{Introduction}

This paper proves both upper and lower bounds
on the mixing time of Glauber dynamics for colourings
of regular trees. It answers in particular the 
question of Hayes, Vera and Vigoda~\cite{HVV}, asking whether the mixing time
of Glauber dynamics
is super-polynomial
for the complete $b$-ary tree with $q=3$ and $b=O(1)$.
We show that the mixing time is not super-polynomial --- it
is $n^{\Theta( b/\log( b))}$.

More generally, we consider Metropolis Glauber dynamics for sampling
proper $q$-colourings of the $n$-vertex complete $b$-ary tree  
when $3\leq q\leq b/2\ln(b)$.
We give both upper and lower bounds on
the mixing time, pinning down the dependance of the mixing time 
on~$n$, $b$ and $q$.
For fixed $q$ and $b$, our upper bound is $n^{O(b/\log b)}$ 
and our lower bound is $n^{\Omega(b/q \log(b))}$, where 
the constants implicit in the $O()$ and $\Omega()$ notation do not depend upon~$n$, $q$ or~$b$.
  
\section{Previous work}

There has been quite a bit of work on Markov chains for sampling the proper 
$q$-colourings of an input graph. 
Much of this work focusses on \emph{Glauber dynamics}, 
which is a general term for a Markov chain which updates 
the colour of one   vertex at
a time.  
Proper colourings correspond
to configurations in the \emph{zero-temperature Potts model} from statistical
physics, and there
is a close connection between the mixing time of Glauber dynamics
and the qualitative properties of the model. In particular,
rapid mixing, specifically $O(n\log(n))$ mixing for an $n$-vertex 
sub-graph of an infinite graph,
often coincides with the case in which the
model has a unique infinite-volume Gibbs measure on the 
infinite graph. 
See Weitz's PhD thesis 
\cite{1048642}
and Martinelli's lecture notes
\cite{martinelli97lectures}
for an exposition of this material.

Martinell, Sinclair and Weitz~\cite{MSW} 
consider Glauber dynamics on the complete $n$-vertex
tree with branching factor~$b$.
They show that for $q\geq b+3$, Glauber dynamics 
for sampling proper $q$-colourings
mixes in $O(n \log(n))$
time for arbitrary boundary conditions.
This result is optimal in the sense that for $q\leq b+2$ there
are boundary conditions for which Glauber dynamics is not even 
ergodic.

It is also interesting to determine whether Glauber dynamics
is rapidly mixing for smaller $q$ in the absence
of boundary conditions.
Hayes, Vera, and Vigoda~\cite{HVV} showed that
there is a $C>0$ such that for all $q > C (b+1)/\log(b+1)$,
Glauber dynamics mixes in polynomial time.
In fact, their result applies to all planar graphs with maximum degree~$b+1$.
They ask in Section~6 whether
the mixing time is super-polynomial for the complete $b$-ary tree with $q=3$ and
$b=O(1)$.
As noted above, we show give upper and lower bounds
showing that the mixing time is polynomial in this case.

As noted in \cite{HVV}, the only previous rapid mixing results for
$q<b+1$ 
were for $3$-colourings of finite subregions of the 2-dimensional
integer lattice~\cite{GMP, LRS} and random graphs~\cite{1182463}.
 
\cite{BVV}
considers \emph{reconstruction} on the complete tree with branching factor~$b$.
They show that for $C=2$ and $q>C (b+1) /\ln(b+1)$
non-reconstruction holds, meaning that,
over random colourings of the leaves, the expected influence on the
root is vanishing.
It is known~\cite{mosselsurvey} that the expected
influence is non-vanishing for a sufficiently large~$q$ satisfying
$q\leq (1-\epsilon) (b+1)/\ln(b+1)$ for
some $\epsilon>0$. This non-vanishing influence implies \cite{BKMP,BVV}
that the mixing time of Glauber dynamics exceeds $O(n \log(n))$.
 
In Theorems~\ref{thm:upper} and~\ref{thm:lower}, we give upper and lower bounds for the mixing 
time for fixed $q$ and $b$
when   $3\leq q\leq b/2\ln(b)$.
Our upper bound is $n^{O(b/\log b)}$ 
and our lower bound is $n^{\Omega(b/q \log(b))}$.

\section{Proof techniques}

The upper bound argument is based on canonical paths.
The lower bound argument is based on conductance.
Essentially, the argument is that it takes a while to move from a colouring
in which the colour of the root is forced to be one colour by
the induced colouring on the leaves
to a colouring in which the colour of the root is forced to be another colour.
This is similar to the \emph{recursive majority} idea  
\cite{BKMP} used to prove a lower bound for the Ising model.

\section{The problem}
 
Fix $b\geq 2$ and fix $q\geq 3$.
Let $[q]=\{0,\ldots,q-1\}$.
Suppose $T$ is a complete $b$-ary tree of height~$H$
--- meaning that there are $H$~edges on a path from the root $r(T)$
to any leaf. 
Let $V$ be the set of vertices of $T$ and $n=|V|$.
Let $L$ be the set of leaves of $T$.
Note that
$$n=\frac{b^{H+1}-1}{b-1}$$ so
\begin{equation}
\label{eq:H}
H = \frac{\log((b-1)n+1)}{\log(b)}-1.
\end{equation}

The \emph{height} $h(v)$ of a vertex $v\in V$ is the number of edges
on a path from~$v$ down to a leaf.
So a vertex $v\in L$ has $h(v)=0$
and $h(r(T))=H$.

For any vertex $v$ of~$T$, $T_v$ denotes the subtree of~$T$ rooted at~$v$.
For any subtree $T_v$, let $V(T_v)$ be
the set of vertices of $T_v$ and let $L(T_v)$ be  the set of leaves.
A \emph{proper $q$-colouring} of $T_v$ is a labelling of the vertices with elements of~$[q]$
such that neighbouring vertices receive different colours.
Let $\Omega(T_v)$ be the set of proper $q$-colourings of $T_v$ and
$\Omega=\Omega(T_{r(T)})$ be the set of proper $q$-colourings of $T$.

For a colouring $x\in\Omega$, let $x(T_v)$ denote
the restriction of $x$ to the vertices in the subtree $T_v$.
Similarly,
for a set $U\subseteq V(T_v)$
and a colouring $x\in \Omega(T_v)$, 
$x(U)$ denotes the restriction of~$x$ to~$U$.

Let $\calM$ be the Metropolis Glauber dynamics for sampling from~$\Omega$.
To move from one colouring to another, this chain selects a vertex $v$ and a colour $c$
uniformly at random. The vertex $v$ is re-coloured with $c$ if and only if this results in a
proper colouring.
If $q\geq 3$ then the set of proper colourings is connected and $\calM$
converges to the uniform distribution on $\Omega$, which we call $\pi$.
The goal is to study the mixing time of $\calM$ as a function of~$n$, $b$
and $q$. Let $P$ be the transition matrix of~$\calM$.
The {\it variation distance} between distributions~$\theta_1$
and~$\theta_2$ on~$\Omega$ is
$$||\theta_1 - \theta_2|| = \frac12 \sum_i |\theta_1(i)-\theta_2(i)|
= \max_{A\subseteq \Omega} |\theta_1(A)-\theta_2(A)|.$$
For a state~$x\in \Omega$, the mixing time of $\calM$ from
starting state~$x$ is
$$\tau_x(\calM,\delta) =
\min\big\{t>0:
||P^{t'}(x,\cdot)-\pi(\cdot)||
\leq \delta \mbox{ for all $t'\geq t$}\big\}.$$
The mixing time of $\calM$ is given by
$$\tau(\calM,\delta) = \max_x
\tau_x(\calM,\delta).$$

Our results are as follows, where $\lg$ denotes the base-$2$
logarithm and $\ln$ denotes the natural logarithm.

\begin{theorem}
\label{thm:upper}
Suppose $q\geq 3$. 
Let $\calM$ be the Metropolis Glauber dynamics for sampling 
proper $q$-colourings of the $n$-vertex complete $b$-ary tree.
Then for fixed $q$ and $b$
the mixing time
$\tau(\calM,1/(2 e))$ is
$n^{O(b/\log(b))}$ where the constant implicit
in the $O()$ notation does not depend upon $n$, $q$ or $b$.
In particular,
$$\tau(\calM,1/(2e)) \leq 3 b q^2 (1+\lg(n)) n^{3+\frac{3 b}{\ln(b)}}.$$
\end{theorem}

\begin{theorem}
\label{thm:lower}
Suppose $3\leq q\leq b/2\ln(b)$.
Let $\calM$ be the Metropolis Glauber dynamics for sampling 
proper $q$-colourings of the $n$-vertex complete $b$-ary tree.
Then 
for fixed $q$ and $b$
the mixing time
$\tau(\calM,1/(2 e))$ is
$n^{\Omega(b/q\log(b))}$ where the constant implicit
in the $\Omega()$ notation does not depend upon $n$, $q$ or $b$.
In particular,
$$\tau(\calM,1/(2e)) \geq
\left(\frac12 - \frac{1}{2e}
\right)\frac{2}{9}
n^{\frac{b-2}{6(q-1)\ln(b)}}
.$$ 
\end{theorem}

\section{Bounds on $H$}

The calculations arising in the derivation of
Theorems~\ref{thm:upper} and~\ref{thm:lower}
involve~$H$. It is clear from Equation~(\ref{eq:H}) that
$H=\Theta(\log(n)/\log(b))$. Since
we give explicit bounds in the statement of 
the theorems, we also require
upper and lower bounds on $H$.
We record these here. Note that the bounds can be improved, but
we prefer to avoid the complication. 

\begin{lemma}
\label{lem:H}
$H+1\leq \lg(n)+1$ and
$H\leq \ln(n)/\ln(b)$.
If $n\geq b^3$ then $H-1 \geq \ln(n)/3\ln(b)$.
\end{lemma}
 
\begin{proof}
For the first upper bound, use Equation~(\ref{eq:H}) to see that
$$H+1 = {\log_b((b-1)n+1)}
\leq \log_b(b n) = 1 + \log_b(n) \leq 1 + \lg(n),
$$
since $n\geq 1$ and $b\geq 2$.
For the second upper bound, note that
$$
H   = \frac{\ln\left(
n\left(b - 1 + \frac1n\right)\right)}{\ln(b)}-1
= 
\frac{\ln(n)}{\ln(b)} - 
\frac{\ln(b)-\ln(b-1+1/n)}{\ln(b)}
\leq \frac{\ln(n)}{\ln(b)}.
$$
Finally, for the lower bound,
note that
$$
H -1  = \frac{\ln\left(
n\left(b - 1 + \frac1n\right)\right)}{\ln(b)}-2
= 
\frac{\ln(n)}{\ln(b)} +
\frac{\ln\left(b - 1 + \frac1n\right)}
{\ln(b)}-2.
$$
Dropping the non-negative middle term, this is at least
$\frac{\ln(n)}{\ln(b)} -2$,
which gives the result since $\ln(n)/3\ln(b)\geq 1$.
\end{proof}

\section{The upper bound}

In this section we prove Theorem~\ref{thm:upper}.
We will use the canonical paths method of Jerrum and Sinclair~\cite{MR1003059}.
Let $\calM'$ be the trivial Markov chain on $\Omega$  that moves from a state~$x$ 
to a new state~$y$ by selecting~$y$ u.a.r. from $\Omega$. Let $P'$ be the transition matrix of $\calM'$.
Clearly, for any $\delta'>0$, $\tau(\calM',\delta')=1$.
we will definine canonical paths between pairs of colourings in $\Omega$. These
canonical paths will constitute what is called an $(\calM,\calM')$-flow.
Then Theorem~\ref{thm:upper} follows from the following
proposition (which is Observation~13 in the expository paper \cite{comparison})
taking $A(f)$ to be the \emph{congestion} of the flow and $c$ to be $1/q$.
The proof of Proposition~\ref{obsmakeodd} combines Diaconis and Saloff Coste's
comparison method \cite{MR1233621} with upper and lower bounds on mixing time \cite{ MR657512, MR1097463, MR1211324}
along lines first proposed by  Randall and Tetali \cite{MR1757972}.
See \cite{comparison} for details.

\begin{proposition} 
\label{obsmakeodd} Suppose that $\calM$ is a reversible ergodic
Markov chain with  transition matrix~$P$ and stationary
distribution~$\pi$ and that $\calM'$ is another reversible ergodic
Markov chain with the same stationary distribution. Suppose that $f$ is a
$(\calM,\calM')$-flow. Let $c = \min_x P(x,x)$, and assume $c>0$.
Then, for any $0<\delta'<1/2$,
$$\tau_x(\calM,\delta) \leq
\max\left\{A(f)
\left[\frac{\tau(\calM',\delta')}{\ln(1/2\delta')}+1\right],\,
   \frac1{2c}\right\}\ln\frac1{\delta \pi(x)}.
$$
\qed
\end{proposition}

For each pair of distinct colourings 
$x,y\in \Omega$
we will construct a path $\gamma_{x,y}$
from~$x$ to~$y$ using transitions of $\calM$.
This gives an $(\calM,\calM')$-flow $f$
with congestion
\begin{align}
\nonumber
A(f)&=\max_{z,w} \frac{1}{\pi(z) P(z,w)} \sum_{x,y: (z,w)\in \gamma_{x,y}}
|\gamma_{x,y}| \pi(x) P'(x,y)\\
&=
\label{eq:flow}
\frac{n q}{|\Omega|}
\max_{z,w} \sum_{x,y: (z,w)\in \gamma_{x,y}}
|\gamma_{x,y}|,
\end{align}
where the maximum is over 
pairs of distinct states $z$ and $w$ in $\Omega$ with $P(z,w)>0$ (hence, $P(z,w)=1/nq$)
and $|\gamma_{x,y}|$ denotes the length of $\gamma_{x,y}$, which is the number of transitions
on the path.
We will prove the following lemma below. 
\begin{lemma}
\label{lem:calc}
The canonical paths  correspond to an $(\calM,\calM')$-flow $f$ with
$A(f) \leq b q (H+1) n^2 9^{b H}$.
\end{lemma}

Theorem~\ref{thm:upper} follows.
Combining Proposition~\ref{obsmakeodd} with $\delta'=1/2e^2$
and Lemma~\ref{lem:calc}, we get
$$\tau_x(\calM,\delta) \leq
b q (H+1) n^2 9^{b H} \left(\frac{1}{2}+1\right)\ln(|\Omega|/\delta).$$
Since $|\Omega|\leq q^n$,
\begin{align*}
\tau_x(\calM,1/(2e)) &\leq
b q (H+1) n^2 9^{b H} \left(\frac{1}{2}+1\right)\ln(2eq^n)
\\
&\leq
(H+1) b q n^2 \frac32 (2+ n \ln(q)) 9^{b H} \\
&\leq
(H+1) b q^2 n^3 3  e^{3 b H}.
\end{align*}
Theorem~\ref{thm:upper} 
then follows by applying the two upper bounds
in Lemma~\ref{lem:H}.
  
\subsection*{Proof of Lemma~\ref{lem:calc}}

\subsubsection*{Defining the canonical paths: a special case}

We start by defining 
paths between colourings~$x$ and~$y$
for the special case in which, for all $v\in V$,
$y(v)=x(v)+1 \pmod q$. 
The sequence of colourings on the path 
is defined to be the sequence of colourings visited by procedure
$\mathit{Cycle}^{+}$ below when it is called with the input~$T$,
which is initially coloured~$x$.

Here is the description of procedure $\mathit{Cycle}^{+}(\That)$,
where $\xhat$ is a global variable, representing the current colouring
of tree $T$, and the input parameter $\That$ may be any of the subtrees~$T_v$.

\begin{enumerate}

   \item Let 
$\That_{1},\ldots,\That_{b}$ be the subtrees rooted 
at the children of $r(\That)$ and let  $S=\big\{i:x(r(\That_{i}))+1\not=x(r(\That))
\pmod q
\big\}$,

   \item For each $i\in S$ do $\mathit{Cycle}^{+}(\That_{i})$.

   \item Recolour the root~$r(\That)$ so that
   $\xhat(r(\That))=x(r(\That))+1\pmod q$.

   \item For each $i\notin S$ do $\mathit{Cycle}^{+}(\That_{i})$.
\end{enumerate}

Since $q\geq3$, we are guaranteed that
$x(r(\That))+1\not=\xhat(r(\That_{i}))\pmod q$,
for all $i$,
after line~2;  this ensures that the root can
be  recoloured in line~3.

\subsubsection*{Analysis of the special case}

Suppose we observe a transition at some point during the
execution of the procedure
$\mathit{Cycle}^{+}(\That)$, in which the colouring $\xhat$
is transformed by adding $1$ to the colour of some vertex $v$ (modulo~$q$).
How many initial colourings~$x(\That)$ (and hence 
how many final colourings~$y(\That)$)
are consistent with this observed transition?

We will let $s(h)$ denote the maximum number of 
consistent initial colourings $x(\That)$, maximised over all
trees $\That$ of height~$h$ and over all possible transitions.
We will compute an upper bound on $s(h)$.

{\bf Case 1:\quad}  Suppose that $v=r(\That)$.

The subtrees $\That_{i}$ with $i\in S$ have already been processed by
the time that the transition takes place, so
$\xhat(T_{i})=y(T_{i})$ for these trees.  The subtrees with
$i\notin S$ are yet to be processed, so for these trees we have
$\xhat(T_{i})=x(T_{i})$.  However, we do not know the set~$S$
from observing the transition from~$\xhat$.  Thus, as many as
$2^{b}$ initial colourings $x(\That)$ may be consistent with the
observed transition from~$\xhat$.

{\bf Case 2:\quad} Otherwise, $v$ is in one of the subtrees $\That_{k}$ 
rooted at one of the children
of $r(\That)$.  Then, by the
argument of Case~1, there are two choices for the initial colouring
$x(T_{i})$ of every subtree with $i\not=k$;  also there
are two possibilities for $x(r(\That))$, since we don't know whether
line~(3) has been executed at the point of the transition.   Then $s(h)$ satisfies the
recurrence $s(h)\leq\max\{2^{b},2^{b}s(h-1)\}$ with initial condition
$s(0)=1$.  Solving the recurrence, we
discover that at most
\begin{equation}
    \label{eq:sBd}
s(h)\leq2^{b h}
\end{equation}
initial colourings~$x(\That)$
are consistent with the observed transition, so
there are at most 
$s(H) \leq 2^{b H}$ initial colourings~$x$ of~$T$ 
consistent with an observed transition 
  of the procedure
$\mathit{Cycle}^{+}(T)$

\subsubsection*{Defining the canonical paths: the general case}

Let $\mathit{Cycle}^{-}$ be defined analogously
to $\mathit{Cycle}^{+}$ but implementing the permutation
of colours that subtracts~1 (modulo $q$) from every colour;
that is, $y(v)=x(v)-1\pmod q$ for all $v\in V$.

Let $F\subset[q]$ be a set of ``forbidden colours'' of size at
most two.
Given $\mathit{Cycle}^{+}$ and $\mathit{Cycle}^{-}$ it is easy
to implement a procedure $\mathit{Cycle}(\That,F)$ that systematically
recolours the tree $\That$ so that the new colour assigned
to $r(\That)$ avoids the forbidden colours~$F$:  simply apply
$\mathit{Cycle}^{+}$ or $\mathit{Cycle}^{-}$ or neither
in order to bring a colour not in~$F$ to the root of~$\That$.
If we observe a transition during the execution of
$\mathit{Cycle}(\That,F)$ we can tell whether it comes from
$\mathit{Cycle}^{+}$ or from $\mathit{Cycle}^{-}$.

The recursive procedure $\mathit{Recolour}$, to be described
presently, provides
a systematic approach to transforming
an arbitrary initial colouring~$x$ to an arbitrary
final colouring~$y$ using single-vertex updates.
In doing so, it defines
canonical paths between arbitrary pairs of proper colourings $x$ and $y$
of~$T$.  
The sequence of colourings on the path $\gamma_{x,y}$ is defined to be the
sequence of colourings visited by procedure $\mathit{Recolour}$
when it is called with the input $T$ (which is initially coloured $x$) and 
with colouring $y$.

Like $\mathit{Cycle}^{+}$,
the procedure $\mathit{Recolour}$
takes an argument $\That$, which is the 
tree which will be recoloured from $x(\That)$ to $y(\That)$.
It also takes the argument $y$.
As before, $\xhat$ is a global variable representing the current colouring of the
tree $T$, which is initially coloured $x$.
Here is the description of procedure
$\mathit{Recolour}(\That,y)$.
 
\begin{enumerate}
    \item  Let 
$\That_1,\ldots, \That_b$ be the subtrees rooted at the children of $r(\That)$.

    \item  For each $i$, $1\leq i\leq b$,
    do $\textit{Cycle}(\That_{i},\{x(r(\That)),y(r(\That))\})$.
    (This step permutes the colours in a subtree, to allow
    the root to be recoloured in the following step.)

    \item  Assign the root $r(\That)$   its final colour~$y(r(\That))$.

    \item  For each $i$, $1\leq i\leq b$,
    do $\mathit{Recolour}(\That_{i},y)$.
\end{enumerate}

\subsubsection*{Analysis of the canonical paths}

Suppose 
we observe a transition at some point during the execution of a procedure call
$\mathit{Recolour}(\That,y)$ when $\That$ has height $h$.
Let
$\numCP(h)$ be an upper bound on the number of 
pairs $(x(\That),y(\That))$ consistent with this transition, maximised over all
trees $\That$ of height~$h$ and over all possible transitions.
Let
$$
\numT(h)=q(q-1)^{(b^{h+1}-1)/(b-1)-1}
$$
be the number of proper colourings of a $b$-ary
tree of height $h$.
Note that $\numCP(H)$ is an upper bound on the number of canonical paths
$\gamma_{x,y}$ using a given transition. 
In order to compute the congestion~$A(f)$ using Equation~(\ref{eq:flow}),
we need to  compute an upper bound
on $\numCP(H)$. We will compute an upper bound
on $\numCP(h)$  by induction on $h$.
The base case is $\numCP(0)=1$. 

Now suppose $h>0$.
Suppose that the transition starts at a colouring $\xhat$
and changes the colour of vertex~$v$ from $\xhat(v)$
to a new colour.

{\bf Case 1:\quad} First, suppose $v=r(\That)$.  
We start by bounding the number of colourings
$x(\That)$ that are consistent with the transition.
From the transition, we know the initial colour of the root, $x(r(\That))$.
For each subtree~$\That_{i}$, we know that the initial
colouring $x(\That_{i})$ can be obtained by permuting the
colours in $\xhat(\That_{i})$.  There are three
possible permutations (corresponding to adding $-1,0$ or $1$
modulo~$q$).  
So the number of possibilities for $x(\That)$ is at most
$3^{b}$.  
Next we bound the number of consistent colourings $y(\That)$.
The colour $y(r(\That))$ is fixed by the transition,
but we know nothing about the colourings of the subtrees $\That_{i}$
beyond the fact that they must be consistent with the root
being coloured $y(r(\That))$.   Thus there are at most
${((q-1)\numT(h-1)/q)}^{b}$ possibilities for~$y(\That)$.
Overall, we have the upper bound
\begin{equation}
    \numCP(h)\leq{(3(q-1)\numT(h-1)/q)}^{b}
    \label{eq:numCProot}
\end{equation}
in the case $v=r(T)$.

{\bf Case 2:\quad} Now suppose $v$ is contained in one of the subtrees~$\That_{k}$.
It could be that the transition under consideration is employed during
Step~2 of $\mathit{Recolour}$ (Type~A), or in Step~4 (Type~B).

{\bf Case 2A:\quad} Consider first pairs of Type~A.  How many pairs
$(x(\That),y(\That))$ of initial and final colourings
may use the transition?  We'll bound this number by considering
separately the pairs $(x(r(\That)),y(r(\That)))$ and
$(x(\That_{i}),y(\That_{i}))$ and multiplying the results.
For the root, $x(r(\That))=\xhat(r(\That))$, while 
there are $q$ possibilities for $y(r(\That))$.  
For $i<k$, there are at most
three possibilities for the colouring $x(\That_{i})$,
and at most $\numT(h-1)$ for $y(\That_{i})$.  For $i>k$,
$x(\That_{i})$ is fixed by the transition, while there are
at most $\numT(h-1)$ possibilities for $y(\That_{i})$.
Now consider the possibilities for~$x(\That_k)$ and~$y(\That_k)$,
starting with $x(\That_k)$.
Given the transition from~$\xhat(v)$ to its new colour  
we can tell whether the instance of
$\textit{Cycle}(\That_{k},\{x(r(\That)),y(r(\That))\})$
is applying $\mathit{Cycle}^{+}$ to~$\That_k$ or $\mathit{Cycle}^{-}$
to~$\That_k$.
In either case,~(\ref{eq:sBd}) guarantees that the number of initial
colourings~$x(\That_k)$ that are consistent with the
transition is at most $2^{b (h-1)}$.
Since the number of possibilities for $y(\That_k)$ is at most $\numT(h-1)$,
the
number for the pair
$(x(\That_{k}),y(\That_{k}))$
is bounded by $2^{b(h-1)}\numT(h-1)$.  This gives an upper bound of
$3^{b}q{(2^{h-1}\numT(h-1))}^{b}$ on the total number of 
pairs
$(x(\That),y(\That))$ such that the given transition is a Type~A transition.

{\bf Case 2B:\quad} Finally, consider pairs of Type~B.
For the root, $x(r(\That))$ is arbitrary, while
$y(r(\That))=\xhat(r(\That))$, so there are $q$~possibilities in all.
For $i<k$, there are at most $\numT(h-1)$
possibilities for the colouring $x(\That_{i})$,
while $y(\That_{i})$ is fixed.  For $i>k$, there are
three possibilities for $x(\That_{i})$, while there are
at most $\numT(h-1)$ possibilities for $y(\That_{i})$.
Inductively, the number of possibilities for the pair
$(x(\That_{k}),y(\That_{k}))$ is $\numCP(h-1)$.
This gives an upper bound of
$3^{b}q\numT(h-1)^{b-1}\numCP(h-1)$ on the total number of 
pairs
$(x(\That),y(\That))$ such that the given transition is a Type~B transition.

{\bf Completing Case 2:\quad}
Summing the bounds on the 
number of 
pairs
$(x(\That),y(\That))$ such that the given transition is a Type~A or Type~B transition
we obtain an upper bound of
\begin{equation}
    \numCP(h)\leq3^{b}q\,\numT(h-1)^{b-1}
    \big[2^{(h-1)b}\numT(h-1)+\numCP(h-1)\big]
    \label{eq:numCPnotRoot}
\end{equation}
on the total number of canonical paths using a given transition
in the case $v\not=r(\That)$.  Notice that~(\ref{eq:numCPnotRoot}) always
dominates~(\ref{eq:numCProot}) since $h\geq 1$.   
Now 
let
$\chi(h)=\numCP(h)/\numT(h)$. Since 
$q^{b-1}\numT(h)=(q-1)^{b}\numT(h-1)^{b}$, we have the recurrence:
\begin{equation}
\label{eq:chi}
\chi(h)\leq \left(\frac{3q}{q-1}\right)^{b}
\big[2^{(h-1)b}+\chi(h-1)\big],
\end{equation}
with initial condition $\chi(0)=q^{-1}$. Now note that
the recurrence (\ref{eq:chi}) satisfies
$\chi(h)\leq 9^{bh}$.

{\bf Completing the Analysis:\quad}  
Let $\CPlength(h)$ be an upper bound on the  number
of updates performed by $\mathit{Recolor}(\That,y)$ when
$\That$ has height $h$. Thus, 
$\CPlength(H)$ is an upper bound on the length of a canonical path $\gamma_{x,y}$.

Now, by Equation~(\ref{eq:flow}),
$$A(f) = \frac{n q}{|\Omega|}
\max_{z,w} \sum_{x,y: (z,w)\in \gamma_{x,y}}
|\gamma_{x,y}|
\leq \CPlength(H)
\frac{n q}{|\Omega|}
P(H)
= 
 \CPlength(H) \,n\, q\, \chi(H),$$
so to prove Lemma~\ref{lem:calc} we need an upper bound on $\chi(h)$.

The subroutine $\mathit{Cycle}$ creates
paths of length $(b^{h+1}-1)/(b-1)$.  The recurrence governing
$\CPlength(h)$ is thus $\CPlength(h)=(b^{h+1}-1)/(b-1)+b\CPlength(h-1)$, with initial
condition $\CPlength(0)=1$.  
Note that $\CPlength(h)\leq (h+1)b^{h+1}$.
This can be verified by induction on~$h$.
For the inductive step,
$$\lambda(h) = \sum_{j=0}^h b^j + b\lambda(h-1)
\leq \sum_{j=0}^h b^j + b h b^h,$$
which is at most $(h+1)b^{h+1}$
since $\sum_{j=0}^h b^j\leq b^{h+1}$ for $b\geq 2$.
Thus
$\lambda(H)\leq (H+1) b^{H+1} \leq b(H+1)n$.
 Putting it all together, the congestion~$A(f)$ is bounded above
by $  qn\chi(H)\CPlength(H)$ which proves Lemma~\ref{lem:calc}.

\section{The lower bound}
 
Suppose
\begin{equation}\label{eq:assumption}
    2q\leq b/\ln(b).
\end{equation}
 
The lower bound proof will use the following 
fact.
\begin{lemma}
\label{lem:b}
If $q\geq 3$ and
$2q \leq b/\ln(b)$ then $b-2\geq 2(q-1)\ln(q-1)$.
\end{lemma}
\begin{proof}
By (\ref{eq:assumption}),
$q-1\leq q \leq b/2\ln(b)$
so
\begin{align*}
2(q-1)\ln(q-1) &\leq \frac{b}{\ln(b)} \ln\left(\frac{b}{2\ln(b)}\right)\\
&= \frac{b}{\ln(b)} 
\left(\ln(b) - \ln(2\ln(b))
\right)
= b - \frac{b\ln(2\ln(b))}{\ln(b)} 
\\
&\leq b-2,
\end{align*}
where the final inequality holds since $q\geq 3$ so $b\geq 6$ 
so $b\geq 2\ln(b)/\ln(2 \ln(b))$.
\end{proof}

Given a colouring $x\in \Omega$, 
define
$$F(x) = \{w\in V \mid
\forall y\in \Omega(T_w) \mbox{ with }
y(L(T_w))=x(L(T_w)) \mbox{ we have }
y(w)=x(w)\}.$$ Informally, $F(v)$ is
the
set of vertices~$w$ of~$T$ whose colour is \emph{forced} by~$x(L(T_w))$.
Our lower bound will be based on a conductance argument
which shows that it takes a while to move from
a colouring~$x$ in which $r(T)$ is forced to be one colour
to a colouring~$y$ in which $r(T)$ is forced to be  another colour.
It is useful to note that $F(x)$ can be defined recursively using
the structure of~$T$. If $w$ is a child of~$v$ we say that
$w$ is $c$-permitting for~$v$ in~$x$ if either $x(w)\neq c$
or $w\not\in F(x)$ (or both).

\begin{observation}
If $h(v)=0$ then $v\in F(x)$. If $h(v)>0$ then
$v\in F(x)$ if and only if, for every colour $c\neq x(v)$,
there is a child $w$ of~$v$ which is 
not $c$-permitting for~$v$ in~$x$.
\label{obs:rec2}
\end{observation}

The recursive definition of $F(x)$ illustrates the connection
between our conductance argument and 
lower-bound arguments based on \emph{recursive majority functions}
\cite{BKMP, mossel}.

Consider a colouring~$x$ chosen uniformly at random from~$\Omega$.
Suppose $v$ is a vertex at height~$h$,
and let $u_{h}=\Pr(v\not\in F(x))$.  Note that the events $v\not\in F(x)$,
with $v$ ranging over all vertices at height~$h$, are independent
and occur with same probability, namely~$u_{h}$.

\begin{lemma}
\label{lemuh}
$u_h\leq 1/b$.
\end{lemma}

\begin{proof}
The proof is by induction on~$h$. Note that $u_0=0$.
For the inductive step, let $v$ be a vertex at height~$h>0$. 
Consider a colouring~$x$ chosen uniformly at random from~$\Omega$.
Fix a colour $c\neq x(v)$ and a child~$w$ of~$v$.
The probability that $x(w)=c$ is $1/(q-1)$. 
To see this, think about constructing the colouring
downwards from the root: Each vertex chooses a colour
uniformly at random from the colours not used by its parent.
Also,
the
probability that $w\in F(x)$ is $1-u_{h-1}$ and this
is independent of the probability that $x(w)=c$.
(The recursive definition of $F(x)$ makes it easy to see that
these events are independent.)
So the probability that $w$ 
is 
$c$-permitting for $v$ in $x$ 
is
$1-(1-u_{h-1})/(q-1)$. These events are independent
for different children~$w$ of~$v$ so the probability that
every child $w$ is 
$c$-permitting for $v$ in $x$ is
 $$
\left(1-\frac{1-u_{h-1}}{q-1}\right)^{b}.
$$

By Observation~\ref{obs:rec2}, the event $v\not\in F(v)$ occurs when 
there exists a colour $c\not=x(v)$
such that every child $w$ 
if $c$-permitting for $v$ in $x$,
so by the union bound:
\begin{align}
u_{h}=\Pr(v\not\in F(x))&\leq
  (q-1)\left(1-\frac{1-u_{h-1}}{q-1}\right)^{b}\notag\\
  &\leq(q-1)\exp\left(-\frac{b(1-u_{h-1})}{q-1}\right)\notag\\
  &\leq (q-1)\exp\left(-\frac{b-1}{q-1}\right)\label{eq:urec1}\\
  &\leq (q-1)b^{-2}\label{eq:urec2}\\
  &\leq b^{-1},\notag
\end{align}
where (\ref{eq:urec1}) applies the induction hypothesis
and (\ref{eq:urec2}) uses assumption~(\ref{eq:assumption}).
 
\end{proof}

Consider a vertex~$v$ of~$T$ with~$h(v)\geq 1$
and a leaf~$\ell$ that is a
descendant of~$v$. Consider $x\in \Omega$. Say that $v$ is
$\ell$-loose in $x$ if there is a $c\neq x(v)$ such that
every child $w$ of $v$, except possibly the one on the path
to~$\ell$, is $c$-permitting for $v$ in $x$.

Let $\Psi_{v,\ell}$ be the probability that $v$ is $\ell$-loose 
in~$x$ when~$x$ is chosen u.a.r.{} from $\Omega$.
Let~$\varepsilon =
(q-1)\exp\left(-\frac{b-2}{q-1}\right)$.

\begin{lemma}
\label{lemPsi}
Consider a vertex~$v$ of~$T$ with $h(v)\geq 1$ and a leaf~$\ell$ that is
a descendant of~$v$.
$\Psi_{v,\ell}\leq
\varepsilon
$.
\end{lemma}

\begin{proof}
 
The calculation very similar to the calculation in the
proof of Lemma~\ref{lemuh}, with $b-1$
replacing~$b$. Let $h=h(v)$. Then 
\begin{align*}
    \Psi_{v,\ell}&\leq(q-1)\left
    (1-\frac{1-u_{h-1}}{q-1}\right)^{b-1}\\
    &\leq(q-1)\exp\left(-\frac{b-2}{q-1}\right),
\end{align*}
where we have used the fact $u_{h-1}\leq b^{-1}$.
\end{proof}

We are now ready to give the lower bound argument.
The \emph{conductance} of a set $S\subseteq \Omega$
is given by
$$\Phi_S(\calM) =
\frac
{
\sum_{x\in S}\sum_{y \in \overline{S}} \pi(x) P(x,y) +
\sum_{x\in\overline{S}}\sum_{y\in S} \pi(x) P(x,y)
}
{2\pi(S)\pi(\overline{S})}.$$
The conductance of $\calM$ is
$\Phi(\calM) = \mathrm{min}_S \Phi_S(\calM)$, where the min
is over all $S\subset \Omega$ with $0<\pi(S)<1$.
The inverse of the conductance of $\calM$ gives a lower
bound on the mixing time of~$\calM$. In particular,
\begin{equation}
\label{eq:cond}
\tau(\calM,1/(2e)) \geq (1/2-1/(2e))/\Phi(\calM).
\end{equation}
Equation~(\ref{eq:cond}) is due to Dyer, Frieze and Jerrum~\cite{DFJ}.
The formulation used here is 
Theorem~17 of the expository paper
\cite{comparison}.

For $c\in[q]$, let
$S_c = \{ x\in \Omega \mid 
(r(T)\in F(x))\wedge
(x(r(T))=c)
\}
$. Let $S_q = \{x \in \Omega \mid r(T)\not\in F(x)\}$.
Clearly, $S_0,\ldots,S_q$ form a partition of $\Omega$.
Let $S=S_0 \cup \cdots \cup S_{\lfloor q/2\rfloor-1}$.
Then  $\Phi(\calM)
\leq \Phi_S(\calM)$.
 
Now by Lemma~\ref{lemuh} we have $0\leq \pi(S_q) \leq 1/b$.
Also, by symmetry, $\pi(S_c)=\pi(S_{c'})$ for $c,c'\in[q]$.
So 
$$\left(1-\frac1b\right) 
\frac{\lfloor q/2\rfloor}{q-1}
\leq \pi(S) \leq 
\frac{\lfloor q/2\rfloor}{q-1}.$$
Since $b\geq 6$ and $q\geq 3$
this gives 
$\tfrac56 \cdot \tfrac12 \leq \pi(S) \leq \tfrac23$,
so 
$\pi(S)\pi(\overline{S})\geq 
\tfrac13\cdot \tfrac23=
\tfrac29$
Thus
$$
\Phi_S(\calM) \leq \frac94
\left(
{
\sum_{x\in S}\sum_{y \in \overline{S}} \pi(x) P(x,y) +
\sum_{x\in\overline{S}}\sum_{y\in S} \pi(x) P(x,y)
}
\right)
,
$$
and by reversibility
\begin{equation}
\label{eq:Phi}
\Phi(\calM) \leq \frac92
\sum_{x\in S}\sum_{y \in \overline{S}} \pi(x) P(x,y) \leq
\frac92
\sum_{x,y} \pi(x) P(x,y),
\end{equation}
where the summation is over
$x$ and $y$ for which 
$r(T)\in F(x)$ and either $r(T)\not\in F(y)$ or
$x(r(T))\neq y(r(T))$.
Note that if~$x$ and~$y$ contribute to the summation in (\ref{eq:Phi})
then since $P(x,y)>0$, they differ on a single vertex.
Since $r(T)\in F(x)$ we cannot move from~$x$ to a proper colouring~$y$ 
by changing the colour of~$r(T)$. Thus the only possibility
is that $r(T)\not\in F(y)$ and $x$ and $y$ differ on a leaf.
Also, given the dynamics, we have $P(x,y) = 1/(n q)$.

\begin{lemma}
\label{lem:boundPhi}
$\Phi(\calM) \leq \frac92 \epsilon^{H-1}$.
\end{lemma}
\begin{proof}

From Equation~(\ref{eq:Phi}) and the discussion above we have
$$
\Phi(\calM) \leq \frac92
\sum_{x,y} \pi(x) P(x,y)$$

where the sum is over all colourings $x$ and $y$ for which
$r(T)\in F(x)$ and $r(T)\not\in F(y)$ and $x$ and $y$ differ
on exactly one leaf, $\ell$.
Letting $c=y(\ell)$,
 we can write
$$
\Phi(\calM) \leq \frac92
\sum_{x\in \Omega} 
\sum_{\ell\in L}
\sum_{c\in[q]}
1_{x,\ell,c} \pi(x) \frac{1}{n q},$$
 
where $1_{x,\ell,c}$ is the indicator for the event
that $r(T)\not\in F(y)$ when $y$ denotes the colouring formed from~$x$
by recolouring leaf~$\ell$ with colour~$c$.
Multiplying by the $q$ possibilities for $c$ and
noting that $\pi(X)=1/|\Omega|$, we get

$$\Phi(\calM)
\leq
\frac92 \frac{1}{|\Omega|}
\frac{1}{n q} q
\sum_{x\in\Omega,\ell\in L} 1_{x,\ell},
$$
where  
$1_{x,\ell}$ is the indicator variable for the event
that there is
a colour~$c$ such that, when 
$y$ is  obtained from~$x$ by changing the colour of
leaf~$\ell$ to~$c$, we have $r(T)\not\in F(y)$.
This event implies that every vertex~$v$ on the path from~$\ell$ to~$r(T)$
is $\ell$-loose in~$x$. When $x$ is chosen uniformly a random
these events are independent and by Lemma~\ref{lemPsi} they all
have probability at most~$\varepsilon$. So
$$
\Phi(\calM) \leq
\frac92 \frac{1}{|\Omega|}
\frac{1}{n }
b^H |\Omega| \epsilon^{H-1},
$$
where $b^H$ is the number of $\ell$ in the summation and $|\Omega|$ is the number of~$x$.

\end{proof}

Theorem~\ref{thm:lower} follows from Lemma~\ref{lem:boundPhi} since, by Equation (\ref{eq:cond}),
the lemma implies
$$\tau(\calM,1/(2e)) \geq (1/2-1/(2e))  \frac{2}{9}
\epsilon^{-(H-1)}.$$
    
Also

\begin{align*}
\epsilon^{-(H-1)}
 &= 
{
\left(
\frac{1}{(q-1)\exp(-(b-2)/(q-1))}
\right)
}^{H-1}\\
&= 
e^{(H-1)
\left(
\frac{b-2}{q-1}-\ln(q-1)
\right)}.
 \\
\end{align*}
Using Lemma~\ref{lem:b},
this is at least
$$
e^{(H-1)
\left(
\frac{b-2}{2(q-1)}
\right)}.$$
Using Lemma~\ref{lem:H},
this is at least
$$e^{\frac{\ln(n)}{3\ln(b)}
\left(
\frac{b-2}{2(q-1)}
\right)} =
n^{\frac{b-2}{6(q-1)\ln(b)}},$$
which gives Theorem~\ref{thm:lower}.

\bibliographystyle{plain}

\begin{thebibliography}{10}

\bibitem{MR657512}
David~J. Aldous.
\newblock Some inequalities for reversible {M}arkov chains.
\newblock {\em J. London Math. Soc. (2)}, 25(3):564--576, 1982.

\bibitem{BKMP}
Noam Berger, Claire Kenyon, Elchanan Mossel, and Yuval Peres.
\newblock Glauber dynamics on trees and hyperbolic graphs.
\newblock {\em Probab. Theory Related Fields}, 131(3):311--340, 2005.

\bibitem{BVV}
Nayantara Bhatnagar, Juan Vera, and Eric Vigoda.
\newblock Reconstruction for colorings on trees, arxiv:0711.3664, 2007.

\bibitem{MR1233621}
Persi Diaconis and Laurent Saloff-Coste.
\newblock Comparison theorems for reversible {M}arkov chains.
\newblock {\em Ann. Appl. Probab.}, 3(3):696--730, 1993.

\bibitem{MR1097463}
Persi Diaconis and Daniel Stroock.
\newblock Geometric bounds for eigenvalues of {M}arkov chains.
\newblock {\em Ann. Appl. Probab.}, 1(1):36--61, 1991.

\bibitem{1182463}
Martin Dyer, Abraham~D. Flaxman, Alan~M. Frieze, and Eric Vigoda.
\newblock Randomly coloring sparse random graphs with fewer colors than the
  maximum degree.
\newblock {\em Random Struct. Algorithms}, 29(4):450--465, 2006.

\bibitem{DFJ}
Martin Dyer, Alan Frieze, and Mark Jerrum.
\newblock On counting independent sets in sparse graphs.
\newblock {\em SIAM J. Comput.}, 31(5):1527--1541, 2002.

\bibitem{comparison}
Martin Dyer, Leslie~Ann Goldberg, Mark Jerrum, and Russell Martin.
\newblock Markov chain comparison.
\newblock {\em Probab. Surv.}, 3:89--111 (electronic), 2006.

\bibitem{GMP}
Leslie~Ann Goldberg, Russell Martin, and Mike Paterson.
\newblock Random sampling of 3-colorings in \&zopf;2.
\newblock {\em Random Struct. Algorithms}, 24(3):279--302, 2004.

\bibitem{HVV}
Thomas~P. Hayes, Juan~C. Vera, and Eric Vigoda.
\newblock Randomly coloring planar graphs with fewer colors than the maximum
  degree.
\newblock In {\em STOC}, pages 450--458, 2007.

\bibitem{LRS}
Michael Luby, Dana Randall, and Alistair Sinclair.
\newblock Markov chain algorithms for planar lattice structures.
\newblock {\em SIAM J. Comput.}, 31(1):167--192, 2001.

\bibitem{martinelli97lectures}
F.~Martinelli.
\newblock Lectures on glauber dynamics for discrete spin models, 1997.
\newblock Lecture Notes in Mathematics, vol 1717, 1998 pp 93-191.

\bibitem{MSW}
Fabio Martinelli, Alistair Sinclair, and Dror Weitz.
\newblock Fast mixing for independent sets, colorings, and other models on
  trees.
\newblock {\em Random Struct. Algorithms}, 31(2):134--172, 2007.

\bibitem{mosselsurvey}
E.~Mossel.
\newblock {\em Survey: Information flow on trees.}, volume~63 of {\em DIMACS
  Ser. Discrete Math. Theoret. Comput. Sci.}, pages 155--170.
\newblock AMS Press, 2004.

\bibitem{mossel}
Elchanan Mossel.
\newblock Recursive reconstruction on periodic trees.
\newblock {\em Random Struct. Algorithms}, 13(1):81--97, 1998.

\bibitem{MR1757972}
Dana Randall and Prasad Tetali.
\newblock Analyzing {G}lauber dynamics by comparison of {M}arkov chains.
\newblock {\em J. Math. Phys.}, 41(3):1598--1615, 2000.
\newblock Probabilistic techniques in equilibrium and nonequilibrium
  statistical physics.

\bibitem{MR1211324}
Alistair Sinclair.
\newblock Improved bounds for mixing rates of {M}arkov chains and
  multicommodity flow.
\newblock {\em Combin. Probab. Comput.}, 1(4):351--370, 1992.

\bibitem{MR1003059}
Alistair Sinclair and Mark Jerrum.
\newblock Approximate counting, uniform generation and rapidly mixing {M}arkov
  chains.
\newblock {\em Inform. and Comput.}, 82(1):93--133, 1989.

\bibitem{1048642}
Dror Weitz.
\newblock {\em Mixing in time and space for discrete spin systems}.
\newblock PhD thesis, U.C. Berkeley, 2004.

\end{thebibliography}

\end{document}